\documentclass[nonacm,sigconf]{acmart}

\usepackage{float}
\usepackage{algorithm}
\usepackage{algpseudocode}
\usepackage{xspace}

\floatname{algorithm}{Procedure}

\newfloat{ruleset}{htbp}{loa}
\floatname{ruleset}{Ruleset}

\newfloat{ordering}{htbp}{loa}
\floatname{ordering}{Ordering}

\newcommand{\now}{\lq now\rq\xspace}

\newtheorem{definition}{Definition}
\newtheorem{problem}{Problem Statement}
\newtheorem{assumption}{Assumption}
\newtheorem{lemma}{Lemma}

\newtheorem{proposition}{Proposition}
\newtheorem{corollary}{Corollary}

\AtBeginDocument{%
  \providecommand\BibTeX{{%
    \normalfont B\kern-0.5em{\scshape i\kern-0.25em b}\kern-0.8em\TeX}}}

\setcopyright{acmcopyright}
\copyrightyear{2021}
\acmYear{2021}
\acmDOI{.................}

\begin{document}
\setlength{\floatsep}{5pt plus 2pt minus 2pt}
\setlength{\textfloatsep}{5pt plus 2pt minus 2pt}

\title{Unbiased Deterministic Total Ordering of Parallel Simulations with Simultaneous Events}

\author{Neil McGlohon}
\affiliation{%
  \institution{Rensselaer Polytechnic Institute}
  \streetaddress{110 8th Street}
  \city{Troy}
  \state{New York}
  \country{USA}
  \postcode{12180}
}
\email{mcglon2@rpi.edu}

\author{Christopher D. Carothers}
\affiliation{%
  \institution{Rensselaer Polytechnic Institute}
  \streetaddress{110 8th Street}
  \city{Troy}
  \state{New York}
  \country{USA}
  \postcode{12180}
}
\email{chris.carothers@gmail.com}

\begin{abstract}
  In the area of discrete event simulation (DES), event simultaneity occurs when any two events are scheduled to happen at the same point in simulated time. Simulation determinism is the expectation that the same semantically configured simulation will be guaranteed to repeatedly reproduce identical results. Since events in DES are the sole mechanism for state change, ensuring consistent real-time event processing order is crucial to maintaining determinism. This is synonymous with finding a consistent total ordering of events.
  
  
  In this work, we extend the concept of virtual time to utilize an arbitrary-length series of tie-breaking values to preserve determinism in parallel, optimistically executed simulations without imposing additional bias influencing the ordering of otherwise incomparable events. Furthermore, by changing the core pseudo-random number generator seed at initialization, different orderings of events incomparable by standard virtual time can be observed, allowing for fair probing of other potential simulation outcomes. We implement and evaluate this extended definition of virtual time in the Rensselaer Optimistic Simulation System (ROSS) with three simulation models and discuss the importance of deterministic event ordering given the existence of event ties.
\end{abstract}



\maketitle

\section{Introduction}\label{sec:introduction}
Discrete event simulation (DES) is an effective method for simulating a wide variety of phenomena. It quantifies any occurrence in the simulation as being an event that occurs on a simulation entity. This event alters the state according to some encoded context and occurs at a temporal coordinate in the simulation defined by an offset from its current definition of time. Because simulation state only changes when an event happens, the outcome of a simulation is entirely dependent on the initial state and the set of ordered events that alters it.

It is expected that two identically configured simulations will produce identical results. Simulations with this property are referred to as deterministic and it is a common expectation for models developed in DES. 
Maintaining determinism in a simulation can help confirm that there is not additional errant, undefined behavior occurring. The defined rules of a deterministic simulation, given a singular input, result in a singular answer.

From a high-level view, maintaining determinism appears to be a simple task: process events in the order in which they are scheduled to occur. However, there are complications that can add complexity to this objective. In a parallel processing environment, a simulation is partitioned across multiple processors, with events occurring between entities that may not exist on the same processor. Because of this, the nature of inter-processor clock drift and individual processing delays, ensuring a deterministic ordering of event execution requires additional synchronization overhead.

Fortunately, this is generally a solved problem in Parallel Discrete Event Simulation (PDES) with various synchronization methods for conservative parallel and optimistic parallel execution. One complication, however, is event simultaneity: events that occur at the same coordinate in virtual time. Given two simultaneous events that each result in different, non-commutative, state changes on the occurring entity, final simulation results will naturally depend on the ordering of these two events. 

In a parallel environment, an inter-processor event may arrive at the receiving processor after an earlier timestamped event had been processed, violating the virtual time-guaranteed happens-before relation of the two events. Correction methods like reverse-computation or checkpoints can be used to revert the simulation to the last known safe state and ensure that a consistent and valid causal ordering is maintained despite this phenomenon. However, there could be a problem when two events are to occur at the same coordinates in virtual space-time. Given the encoded standard virtual timestamps alone, there is not a clear way to define which event happens-before the other which results in a loss of determinism.

To regain assured determinism, a set of rules can be defined dictating the expected ordering of two simultaneous events. Given a strictly increasing event counter on each processor that increments with every new event, then in the case of simultaneous events, we can give priority based on the sending processor ID. If that, too, is identical, then we can give priority based on the event ID encoded from the event counter. Determinism was regained but with one major caveat: the resulting ordering is influenced by the bias of the chosen of rule-set.

Because of this bias imparted by the ruleset, it is difficult to know for certain if the resulting simulation behavior is exemplary of the simulated model's behavior as a whole or a special case determined by said ruleset's influence on event ordering. In this work we extend the definition of virtual time to include a set of tie-breaking values and devise a mechanism for unbiased arbitration of simultaneous events in a parallel discrete event simulator with optimistic execution. This is accomplished by utilizing a rollback-safe tie-breaking uniform random number generator. We develop this mechanism into the Rensselaer Optimistic Simulation System (ROSS) and exhibit its effects on simulation performance and model behavior. Because ROSS, even in optimistic execution, utilizes fully deterministic pseudo-random number generators (PRNGs), different RNG initialization seed values will yield different but deterministic event orderings, even in the presence of simultaneous events, allowing for deeper statistical analysis of simulated models.

In this work we propose and argue three new mechanisms for finding a deterministic total ordering of events in a parallel discrete event simulation with simultaneous events:

\begin{enumerate}
  \item Unbiased random total ordering of otherwise incomparable simultaneous events.
  \item Biased random total ordering of otherwise incomparable simultaneous events, including events created with zero-delay/offset.
  \item Unbiased random total ordering of otherwise incomparable simultaneous events, including events created with zero-delay/offset.
\end{enumerate}



\section{Problem Definition}
For this work, we specifically utilize the Rensselaer Optimistic Simulation System (ROSS) PDES engine~\cite{carothers2000,carothers2002} which implements the Time Warp protocol~\cite{jefferson1985fast,jefferson1987,jefferson1987time} in conjunction with virtual time~\cite{jefferson1985}. In a ROSS simulation, entities or agents are represented as Logical Processes (LPs). These LPs are mapped to the various Processing Elements (PEs) that may exist. Typically, there is one PE per physical MPI process that is participating in the actual execution of the simulation.

ROSS is capable of being run in three main modes of synchronization. When the simulation is to be executed on a single process, it is executed in \emph{sequential} mode. Maintaining determinism in this mode is trivial. The challenge of determinism comes when additional PEs are added to the simulation. With additional PEs comes the advent of \emph{remote} events, those whose destination is on a different PE than its source. 


All propositions in this work are based on the following assumptions and definitions:


\begin{assumption}\label{as:inter-ordering}
  In a parallel environment, arrival order of inter-processor messages is not guaranteed to be consistent across executions of a simulation.
\end{assumption}

This is due to natural differences of local clock-speed. Even the smallest perturbation can lead to enough clock drift to change the order in which any two events from different processors may arrive at a third processor. If this is not accounted for, it can lead to drastically different final results. Fortunately, this is addressed by Time Warp's virtual time approach to parallel event processing.

\begin{assumption}\label{as:timestamp}
  Every event in the simulation is encoded with a standard virtual timestamp specifying when in simulation time the event occurs.
\end{assumption}

Without virtual time, a receiving process will have no ability to determine when in the future the event should occur. Given this encoded virtual time, it is possible for the receiving process to deduce a partial ordering of this received event and others that the process is aware of.

\begin{assumption}\label{as:pe-queues}
  Each processing element in a PDES system has a priority queue maintaining a proper order of arrived events to be processed.
\end{assumption}

Because different LPs across PEs may be scheduling events at various times in the future, we cannot assume that events will arrive in the exact order that they are to be executed -- and by Assumption~\ref{as:inter-ordering}, can not be guaranteed even if one tried to make it so. Thus, each PE should manage a queue maintaining a proper order in which events should be processed. This is almost certain to exist in any parallel simulator that implements virtual time.



In an optimistically executed simulation events that arrive at an LP after its current definition of \now are referred to as \emph{straggler} events~\cite{fujimoto1990}. To address this, optimistic simulations can be rolled back. After rolling back the simulation to a known safe state, the undone events can be correctly ordered in the priority queue and re-processed. None of this can be accomplished without some ordered queue.

\begin{assumption}\label{as:no-past-events}
  No LP can create events at a virtual time in the past relative to its own definition of virtual time.
\end{assumption}

This is a safe assumption as the entirety of discrete event simulation relies on the basis of maintaining causal-order; causal ordering can be formally stated as:


\begin{definition}[from~\cite{jefferson1985}]\label{def:causal}
  Event $A$ causes $B$ ($A\rightarrow B$) if there exists any sequence of events $A=E_0,E_1,\ldots,E_n=B$ such that for each pair $E_i$ and $E_{i+1}$ of adjacent events either (a) $E_i$ and $E_{i+1}$ are both in the same process and the virtual time of $E_i < E_{i+1}$ or (b) event $E_i$ sends a message to be received at event $E_{i+1}$.
\end{definition}




Because, in virtual time simulations, the virtual space-time coordinates of an event define the happens-before and causal relations of events, the only instance where two events are considered \emph{incomparable} is when they have identical virtual timestamps, i.e. simultaneous events or an \emph{event tie}. 

\begin{definition}\label{def:comparable}
  Two events are considered comparable if a consistent happens-before relation can be inferred between the two. If no happens-before relation can be inferred, then they are considered incomparable.
\end{definition}

\begin{definition}\label{def:zero-offset}
 An event created with zero virtual-time delay between its encoded timestamp and the creating LP's definition of \now (timestamp of the causal event) is referred to as a zero-offset event.
\end{definition}

\begin{assumption}\label{as:rng}
  Every LP has at least one, rollback-safe, pseudo-random number generator stream solely dedicated to generating uniform random values to encode into events that it creates.
\end{assumption}

To achieve the goal of creating an unbiased ordering of simultaneous events via a pseudo-random number generator, it is imperative that this stream be completely independent from other streams accessed in the model and deterministically rolled back. For reasons which will be explained in detail in Section~\ref{sec:random-determinism}, it is also critical that this PRNG can be deterministically rolled back to a previous state should the events it generated values for be cancelled or rolled back.

With the above assumptions and definitions, we posit the following problem:

\begin{problem}\label{prob:deterministic-ties}
  Given an optimistically or conservatively executed parallel discrete event simulation with the above assumptions, consisting of a partially ordered set of virtual-time encoded events $E$ which may or may not contain event ties, find a mechanism which ensures unbiased total ordering of $E$ such that any subsequent executions of the simulation are deterministic.
\end{problem}

\section{Event Simultaneity}\label{sec:simultaneity}
Another way to phrase Problem~\ref{prob:deterministic-ties} is: given a simulation of events, find a mechanism to randomly establish the partial ordering found by a parallel simulation in a way that is deterministic and identical to the total ordered version found by a sequential execution of the same simulation. If all events in a partially ordered simulation are comparable to each other, then the partially ordered events are also considered totally ordered.



As briefly mentioned in Section~\ref{sec:introduction}, deterministic ordering of simultaneous events can be achieved by establishing a ruleset that is enforced by each PE dictating which event happens-before another. If this ruleset is precise enough to render every possible event in the simulation comparable to every other event then a deterministic total ordering can be established from its partial counterpart. In Figure~\ref{fig:simultaneous}, events $A,B,C$ and $D$ all occur simultaneously in the simulation. How this ruleset is defined will specify how these events are ordered to occur in the simulation. The orderings of $C$ and $D$ are specifically of consequence due to them each operating on the same LP.

\subsection{Ordering Bias}
This ruleset would, in effect, turn incomparable, simultaneous, events into comparable events that occur \emph{simultaneously yet infinitesimally before or after one another}. An example ruleset would be as follows, in decreasing order of importance, cascading down to break any subsequent value ties.

\begin{ruleset}
\caption{Ordering with Explicit Bias}
\flushleft 
Give priority to:
\begin{enumerate}
  \item Events with lower virtual timestamp
  \item Events with lower PE ID
  \item Events with lower LP ID
  \item Events with lower (per PE) event ID
\end{enumerate}
\end{ruleset}

This however will result in a single final result with likely few options that would allow users to explore alternate orderings. Perhaps with greater consequence is the amount of bias that this ruleset instills into the simulation. In the case of a timestamp tie, it will always give priority to whatever event came from a lower numbered processor. This means that LPs mapped to lower numbered processors, for instance, will always take precedence should they tie and that the same simulation but with different numbers of PEs or a different mapping will likely generate different final results.

Given the possibility of a two-way event tie, it is possible to relinquish the simulated model from the influences of our bias by instead relying on a fair coin flip.

\begin{figure}[t]
  \centering
  \includegraphics[width=.8\linewidth]{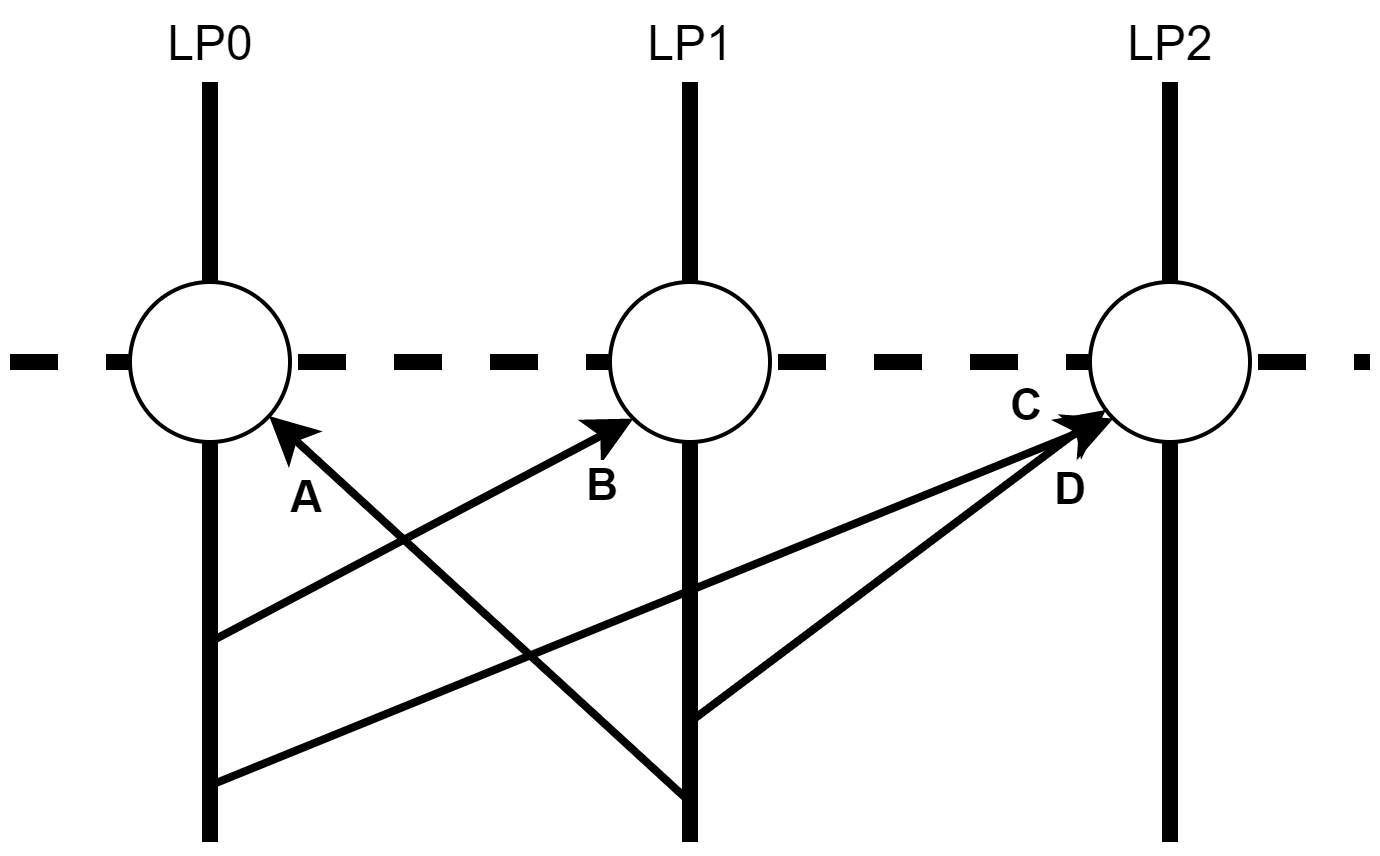}
  \caption{In this LP-Event diagram, all events: $A,B,C,$ and $D$ occur at each of their respective LPs simultaneously. Determining how these events, $C$ and $D$ in particular, should be ordered in a final total ordering may have an effect on final simulation results.}\label{fig:simultaneous}
\end{figure}

\begin{ruleset}
  \caption{Ordering without Bias given 2-way event ties}
\flushleft
Give priority to:
\begin{enumerate}
  \item Events with lower virtual timestamp
  \item Events that win a coin flip
\end{enumerate}
\end{ruleset}

With this rulset, given two tied events, A and B, there is a 50\% chance of being ordered as [A,B] and a 50\% chance of being ordered as [B,A]. If we can randomly generate the values necessary for this coin flip to occur, then we can, without bias, order all events in a simulation with at most 2-way event ties.

This can easily be extended to $n$-way ties. Instead of a coin-flip, encode an i.i.d. uniform random (UR) value into an event when it is created and implement the following ruleset.

\begin{ruleset}
  \caption{Ordering without Bias given $n$-way event ties} \label{rule:n-way}
\flushleft
Give priority to:
\begin{enumerate}
  \item Events with lower virtual timestamp
  \item Events with lower tie-breaking value
\end{enumerate}
\end{ruleset}

Given an $n$-way tie, each event having its own independently generated uniform random value, any specific ordering of these events constitutes a single possibility (occurring with probability $1/n!$) out of all possible permutations of these tied events. Furthermore, the probability of any two tied events being ordered a specific way with respect to each other is 50\%. This method is equivalent to what is described and proven in~\cite{cormen2001} as \texttt{PERMUTE-BY-SORTING} yielding the following lemma:

\begin{lemma}[From~\cite{cormen2001}]\label{lem:permutation}
  Given a list of length $N$ with $N$ distinct i.i.d. uniform random values, a uniform random permutation can be found by assigning one value to each item in the list and sorting according to these values.
\end{lemma}

If the probability of tie-breaker collision is negligible, then it is possible to generate a uniform, or unbiased, random permutation by sorting based on the uniform randomly generated key. It follows that by using Ruleset~\ref{rule:n-way}, we can extract an unbiased random ordering of an $n$-way event tie.

One may question whether any random ordering of these events is safe. The most important thing in ordering simultaneous events is to uphold event causality: no event caused by another can happen before the other event.

\begin{lemma}\label{lem:safe-causality}
  If there are no zero-offset events in a simulation, then no events participating in a tie caused any other events in said tie.
\end{lemma}
\begin{proof}
  Given the contrapositive: if events participating in the tie did cause other events in the tie, then there are zero-offset events. Suppose an event participating in the tie caused another event in the tie. Because both events are participating in the tie, the virtual time difference between the two would be zero. By definition, the caused event was scheduled with zero-offset.
  Thus, by contrapositive proof, if zero-offset events are not allowed, then no events in an event tie caused any other events in the same tie.
\end{proof}

A consequence of this lemma is that if there are no zero-offset events in a simulation, then because there is no causal linking between any events participating in the tie, no two events in the tie are ordinarily comparable. If they were comparable by their virtual timestamps, then they would not be in this tie in the first place.

Because events that are not causally related to each other and incomparable could be processed in any order without violating causality, then a random permutation of these events constitutes a valid ordering.


\begin{proposition}\label{prop:no-zero-unbiased}
  Providing, in addition to the standard virtual timestamp, a secondary uniform random value and comparing tied events based on that will yield an unbiased ordering of events without violating causality given no zero-offset events.
\end{proposition}
\begin{proof}
  Lemma~\ref{lem:safe-causality} assures that no events in any event tie will have a causal relationship with any others in the same tie. Because these events, given their standard virtual timestamp, are thus incomparable, any possible permutation of these events is causally safe and by Lemma~\ref{lem:permutation}, an implementation of Ruleset~\ref{rule:n-way} will allow these events to be comparable and yield an unbiased ordering of events.
\end{proof}

For clarification, we define the combination of a virtual timestamp and any secondary values determining event priority as a \emph{virtual time signature}. Events created in accordance to Proposition~\ref{prop:no-zero-unbiased} that are incomparable by their virtual timestamps are comparable by their virtual time signature.

It is important to note, however, that Lemma~\ref{lem:permutation} relies on \emph{distinct} values from its uniform random number generator. By the nature of PRNGs, this is not possible. For the purposes of our work, though, we utilize sufficiently precise PRNGs with very large periods and thus we assume the probability of generating non-distinct values is negligible. 

Should this probability still be considered too high, additional tie-breaking values can be generated and encoded via supplementary, independent, PRNG streams. By applying \texttt{PERMUTE-BY-SORTING} recursively to sub-lists containing PRNG value collisions. This effectively increases the bit precision of the primary tie-breaking value, making overall collision even less likely while maintaining the expected lack of bias.

\subsection{Determinism with Randomness} \label{sec:random-determinism}
We have now established that given a good choice of ruleset, we can create \emph{an} unbiased ordering of events in a simulation. But this alone does not grant determinism. The ROSS simulator guarantees Assumption~\ref{as:rng}, allocating an independent PRNG stream on every LP with the sole duty of generating the tiebreaker values encoded into each event that that LP creates.

One valuable property of PRNGs is that given a specific seed, the sequence of values generated from it is deterministic. In a sequentially or conservative parallel executed simulation, this allows for pseudo-randomness without sacrificing determinism. In neither of these execution modes is there ever a rollback to a previous simulation state.

When executing a parallel simulation optimistically, however, rollbacks are expected. If the state of a PRNG stream, indicating the position of the \lq next\rq~value in the sequence, is not also rolled back when necessary, then a different set of numbers will be generated for events when the simulation is resumed. As a result, LP state changes and decisions based on that PRNG stream will differ from an execution where said rollbacks did not occur -- a loss of determinism.


\begin{corollary}
  Using a rollback-safe tie-breaking uniform random value in execution of Proposition~\ref{prop:no-zero-unbiased} will yield a deterministic unbiased ordering of events without violating causality given no zero-offset events.
\end{corollary}


Any extra identifier that enables the comparison of two events that would otherwise be incomparable must be implemented into all parts of a system that compares two events. If, for instance, a PDES system is configured to process events in a specific order but care is not taken to recognize any straggler-message -- based on virtual time signatures -- and correct it, then there will no longer be any guarantee of determinism. Any events irrevocably committed to the simulation history must be in agreement with the new, extended, definition of virtual time.

\section{Zero-Offset Event Simultaneity}\label{sec:relaxing}
Proposition~\ref{prop:no-zero-unbiased} makes certain assumptions about the nature of the problem and the simulation environment it operates in. Specifically, it does not allow for zero-offset events to exist in the simulation. While this may seem rather inconsequential from a real-world perspective -- we are quite used to imposed speed limits in traversing space-time -- this is not necessarily a guaranteed assumption for every possible simulated model.

To allow for zero-offset events in a simulation, there are some finer details that must be established to ensure that a deterministic, valid, unbiased total ordering of events can be found. Part of what makes this difficult is that the act of creating zero-offset events, by definition, creates event ties. As a result, an event that generates a new event with a zero-offset is creating a new event that, in virtual time, occurs at the same time as the event that created it. In this instance, the order in which these two events are committed to simulation history is very important.


It's also important to note that while Assumption~\ref{as:no-past-events} is standard for virtual timestamps, the same must also be held true for any happens-before relations inferred from an extended definition of virtual time. Furthermore, the virtual time signature of any event must be strictly greater than the time signature of its causal parent.

\subsection{Challenge of Zero-Offset Events} \label{sec:zero-offset-challenge}
We have shown with Proposition~\ref{prop:no-zero-unbiased} that an unbiased total ordering of simultaneous events can be found with a uniform random value acting as a tie-breaker. This works by establishing an unbiased hard rule of how any two events should be ordered in the simulation. Proposition~\ref{prop:no-zero-unbiased} worked because any arbitrary ordering of otherwise incomparable events will result in a valid simulation.

This becomes tricky once zero-offset events are introduced. Because a parent event and its zero-offset child have the same virtual timestamps but one event \emph{definitely} caused the other, the ordering in which these two events should be ordered in the final simulation history cannot be arbitrary. 

In a solution using Proposition~\ref{prop:no-zero-unbiased}, to create an unbiased random ordering of events, an absolute comparison of i.i.d. random tie-breaking values is utilized. Given an event $E$ that generates a child event $E'$ with zero-offset, it is entirely possible that the uniform random tie-breaking value of $E'$ could be less than that of event $E$. By our extended definition of virtual time, this would imply that event $E'$ \emph{happened before} event $E$ which is impossible since event $E$ directly caused $E'$. This contradiction caused by a zero-offset event is exactly the same as what would happen if a time-travelling negative-offset event were created. This establishes the following:

\begin{lemma}\label{lem:causal-guarantee}
  No event can be created in a way that would classify it as happening before any events that caused or happened-before it.
\end{lemma}

Without the tie-breaker value enforcing the final ordering of these events, this will execute and complete fine - but at the risk of non-determinism if there are other event-ties. But with the tie-breaking mechanism, the PDES engine will enter an endless loop of rollbacks and replays as it tries, fruitlessly, to reconcile this situation and adhere to the deterministic tie-breaking order.

If all we consider was their standard virtual timestamps, the two events are incomparable but they \emph{are} comparable by the ordered chain of causality. Definition~\ref{def:causal} does not account for zero-offset events; if we are to have a stable simulation with zero-offset events we need to extend our definition of virtual time to ensure that the virtual time signature of $E < E'$ and so on for any descendants of event $E'$.

\begin{definition}
  An event which shares a standard virtual timestamp with another but has a lower tie-breaking value is defined to happen infinitesimally before the other. The converse implies that the other event happens infinitesimally after the first.
\end{definition}

\begin{definition}
  A zero-offset event, sharing the virtual timestamp with that of its parent, is considered to happen infinitesimally after its parent. The converse implies that the parent happens infinitesimally before its child(ren).
\end{definition}

To demonstrate the challenge of utilizing the tie-breaking mechanism we have established in Proposition~\ref{prop:no-zero-unbiased}, let us assume that we have some ability to ensure that Lemma~\ref{lem:causal-guarantee} is maintained while using the same tie-breaking mechanism as before. Then we have nothing to prevent the possibility of generating events with virtual time signatures to \emph{happen infinitesimally before} others that have already been processed.

\begin{table}[t]
  \caption{Events and their virtual times and tie-breaking values for the example shown in Figure~\ref{fig:sequential-broken}.}\label{tab:event-timestamps}
\begin{tabular}{|r|c|}
\hline
Event&(time, tie-breaker)\\\hline
$A$&$(1,0.1)$\\\hline
$B$&$(1,0.15)$\\\hline
$A'$&$(1,0.40)$\\\hline
$B'$&$(1,0.30)$\\\hline
$A''$&$(1,0.20)$\\
\hline
\end{tabular}
\end{table}

\begin{figure}
  \centering
  \includegraphics[width=.8\linewidth]{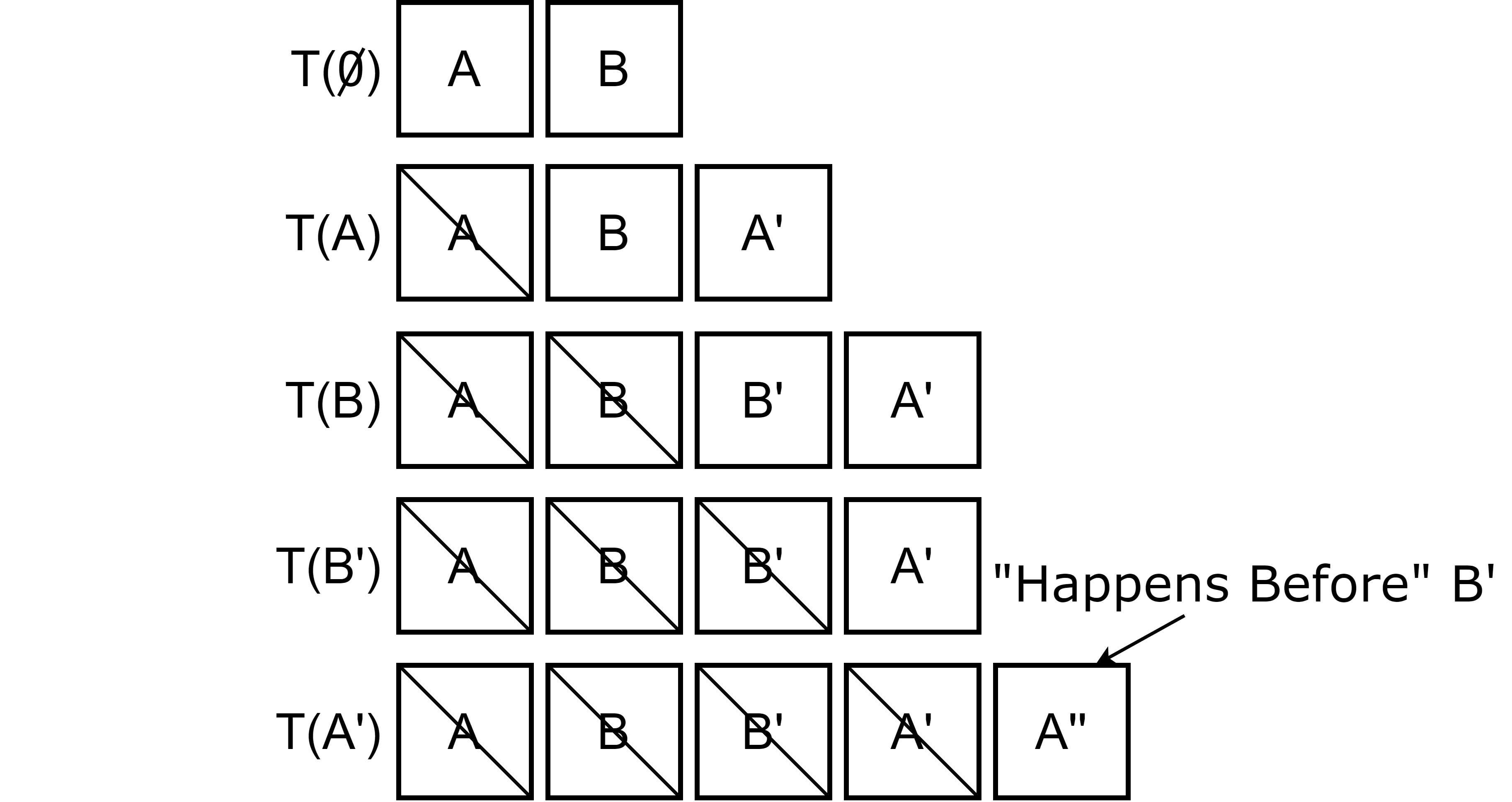}
  \caption{Diagram showing event processing queue state while processing previously given simultaneous events. $T(-)$ represents the time in the simulation after a given event has been processed. Event $A''$ is created by $A'$ but according to its tie-breaker value should happen before $B'$ which has already been processed as designated by a slash.}\label{fig:sequential-broken}
\end{figure}

Consider a sequential simulation, there is no ability to roll-back events should a new one be generated that \emph{happens before} one previously processed -- if a mechanism cannot be used for generating a valid ordering in sequential, then it is unlikely to consistently work for parallel executions either. As an example, two events, $A$ and $B$ each generate a zero-offset child event: $A'$ and $B'$. $A'$ generates a single zero-offset child, $A''$. The events are each scheduled to occur at the same point in virtual time but have uniform random generated tie-breaking values to determine the order of events not directly caused by each other; the virtual time and tie-breaker values for these events are shown in Table~\ref{tab:event-timestamps}.

In this sequential simulation, we would start by processing events $A$ and then $B$. This would schedule two new zero-offset events $A'$, and $B'$. We could process these two in order according to their tie-breaker values without issue ($B'$ first). However, after processing $A'$, another child event is created, with a tie-breaker value that is lower than $B'$. The simulation processing queue state for this example can be observed in Figure~\ref{fig:sequential-broken}. The result of this is an irreparable error in the simulation. In a sequential execution, an effective negative-time event was received, breaking Lemma~\ref{lem:causal-guarantee}. In an optimistic execution, we will enter an infinite rollback scenario and cause instability in sequential.

This is exactly the challenge of creating a fair deterministic ordering of zero-offset events: it is not possible for the simulator to know if a zero-offset event will create another zero-offset event until after its processed. Thus, how can a fair random ordering of zero-offset tie events be determined?

\subsection{Biased Random Causal Ordering}\label{sec:biased-random}
It is now clear that a single pair of randomly defined tie-breaking values is insufficient for ordering two events given the existence of zero-offset events. It is always possible to violate causality and that should \emph{never} be possible; any scheme to determine ordering two events must guarantee causality is maintained.

We know, however, that a simulator that utilizes standard virtual time can generate a deterministic total ordering of events if no ties exist because all events have a unique time and the final ordering should just be monotonically increasing by their timestamp. This is possible because when new events are created, they are encoded with an offset: an additive \lq time from now\rq~value. So long as that value is not negative, it guarantees causality is maintained; as long as the value is positive, it guarantees determinism.

What if, when generating a zero-offset event, instead of using a random tie-breaking value from (0,1), we generate that random value but \emph{add} it to the random tie-breaking value of the parent event. The new tie-breaking value will be strictly greater than the parent, assuring that any comparison of the two will always order the parent first. A LP-Event diagram using the same generated values in the previous example using Table~\ref{tab:event-timestamps} but adding consecutively generated tie-breaking values when creating child events is shown in Figure~\ref{fig:additive-diagram}. 

\begin{figure}[t]
  \centering
  \includegraphics[width=.8\linewidth]{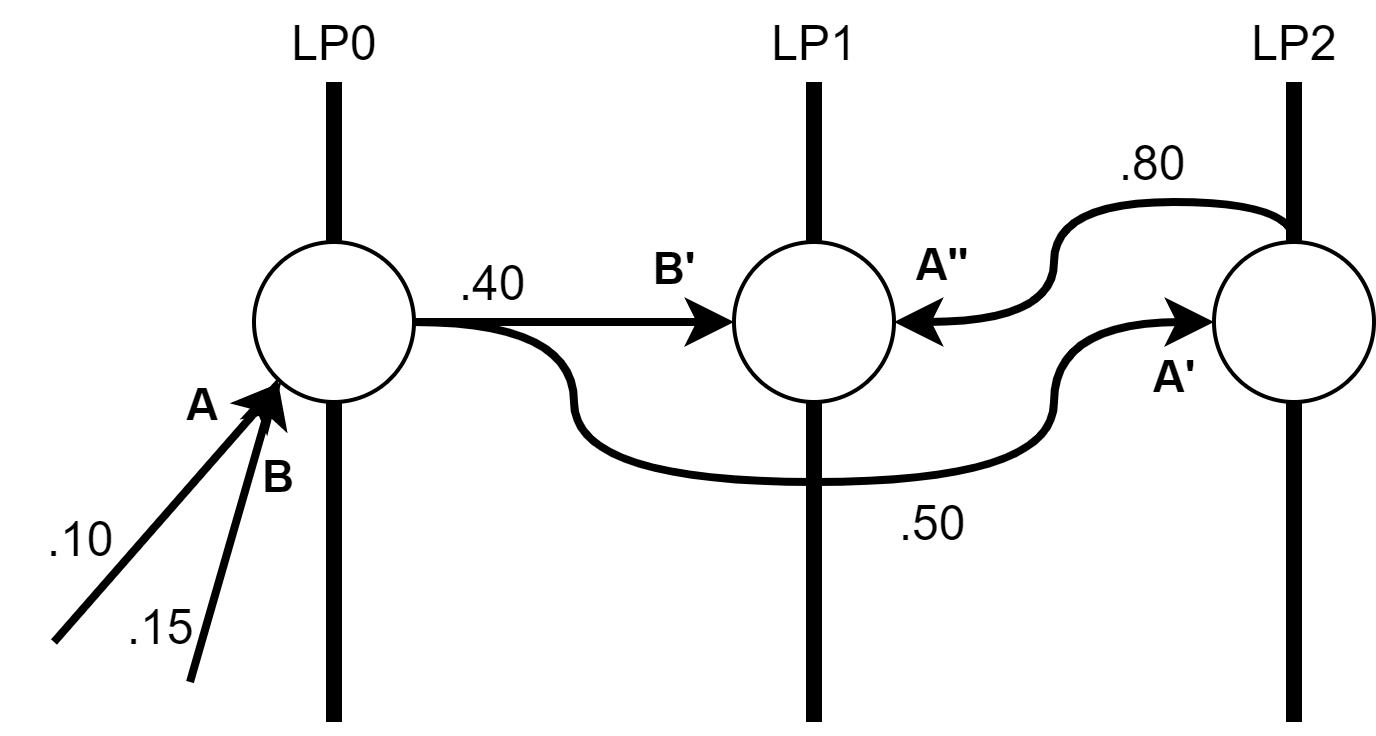}
  \caption{LP-Event Diagram depicting events all simultaneously occurring with the same virtual timestamp but depicted with varying legal tie-breaking values encoded using the additive scheme described in Section~\ref{sec:biased-random}. The final ordering in the simulation for these tied events is: $[A,B,B',A',A'']$}\label{fig:additive-diagram}
\end{figure}

\begin{proposition}\label{prop:biased-zero-offset}
  Providing, in addition to the standard virtual timestamp, a secondary value defined by summing deterministic uniform random values generated by causally related zero-offset events and comparing tied events based on that will yield a randomized-but-biased deterministic ordering of events, including zero-offset events, without violating causality.
\end{proposition}
\begin{proof}
  Given four simultaneous events $A,B,C,$ and $D$ where $A\rightarrow~B$~ and ~$A\rightarrow~C\rightarrow~D$ and $B,C,D$ are zero-offset from $A$ with additive tie-breaker values.
  Assume that an event $D$ is an event that violates Lemma~\ref{lem:causal-guarantee} in conflict with a previously processed event $B$. The tiebreaker value of $D$ is strictly greater than that of its parent: $C$. The only way for $D$ to be created so that it happens before $B$ is if $C$ also happened before $B$. If $C$ happened before $B$, then no ordering of $B$ and $D$ would violate Lemma~\ref{lem:causal-guarantee}. Thus, by contradiction, additive tie-breaking values will safely allow for zero-offset event tie-breaking without violating causality established from Lemma~\ref{lem:causal-guarantee}.
\end{proof}



This does allow for randomness to be introduced in the ordering of events but it introduces some potentially unintended bias into the comparison of tie-breaking values. For example, given two events scheduled for the same time: $A$ and $B$. Event $A$ creates a zero offset child $A'$ who creates a zero offset child $A''$. The four events are each scheduled for the same virtual time. If the tiebreaker value for descendants is additively generated, then the probability that event $A''$ comes before event $B$ is not 50\%. it is possible that, based on the tie-breaking value, that $A''$ happens before $B$, but it is far more likely that $B$ happens before $A''$ because the expectation for its tie-breaker value is lower than that of $A''$. 

This is because the tie-breaking value encoded into events $A'$ and $A''$ are no longer a uniform random value from (0,1) but is instead sampled from an Irwin-Hall distribution~\cite{johnson1995continuous}. The expected value of an Irwin-Hall value is proportional to the number of uniform random values added to generate it and thus the more zero-offset descendants recursively created, the less likely they are to happen before another, independent, event. In a way, this may seem natural, but it is important to remember that these events are all scheduled by the model to happen at the same time. 

This method will assuredly result in a deterministic ordering and this ordering will be somewhat randomized based on the simulation's input seed. Is there, however, a fairer way to break ties within zero-offset events?

\subsection{Unbiased Random Causal Ordering}\label{sec:unbiased-causal}
When trying to find a fair way to order simultaneous events, something to consider is whether it even makes sense for other events to interject between an event and its zero-offset child. The two related events are supposed to occur simultaneously with the exception of one being causally dependent by the other (the child \emph{happens infintesimally afterward}). By that logic, it makes sense that they should be consecutive in the final ordering. 

Previously, in Section~\ref{sec:biased-random}, we observed that the only way for an event's child to precede another, independent, simultaneous event if the original event also preceded the other.

What made the additive tie-breaker value successful was that it guaranteed that a zero-offset child event will have a value strictly greater than its parent. But it also established a unique timestamp that will not -- should there be more descendants -- violate some established order of other already processed events.

We can devise another mechanism that will ensure that any subsequent zero-offset children of a parent event be processed before other independent events (with a greater tie-breaking value than said parent). Consider ordering two sequences of numbers. There are numerous ways to pick some relation and choose which should be considered first in an ordering of the two sets. One common way is through \emph{lexicographical ordering}. Just as one would find where to place new words into the Oxford dictionary of English~\cite{stevenson2010oxford}, we start by comparing the first items in the two sets. Should those two match, we recursively compare each of the subsequent items in the sequences until we find a pair that are different. Determining which of the two is \lq less\rq~ than the other, lexicographically, is then based on that final comparison.

For example, let us consider two sequences $S_1 = [5,3,9,3]$ and $S_2 = [5,3,4,6]$. Lexicographically, $S_2 < S_1$ because in the third component-wise comparison, $4 < 9$. We also observe that it does not matter how long the sequence of $S_2$ is; as long as those first three numbers remain in each sequence, then $S_2 < S_1$ \emph{always}.

We can extend our definition of virtual time again to include, not a single tie-breaking value, but a \emph{sequence} of tie-breaking values. The sequence is comprised of the tie-breaking values of historical zero-offset parental events -- whenever a new event is created, a new tie-breaking value is generated and appended to the back of this running list. When a regular-offset event is created, the sequence is discarded and then the newly generated tie-breaking value for the current event is added as the sole value.

When comparing two tie-breaking sequences of different lengths but all $N$ components of the shorter one match the first $N$ components of the longer sequence, priority is given to the shorter one as this can only happen if the event owning the shorter sequence \emph{caused} the other. Referring back to the English dictionary analogy, the single-letter word \lq A\rq~ comes before \lq aardvark\rq~ but \lq aardvark\rq~ comes before \lq apple\rq~\cite{stevenson2010oxford}.

The result of this extension allows for us to make every single event in the simulation comparable to one another via the new virtual time signature. In this case the virtual time signature is defined as the sequence of the virtual timestamp followed by the regular-offset tie-breaker and all consecutively generated zero-offset tie-breakers. Because of the benefits of lexicographical ordering, causality is guaranteed as any zero-offset descendant's time signature will be strictly greater than its parent, grandparent, etc. The happens-before comparison of any two tied events that are unrelated will, as in Proposition~\ref{prop:no-zero-unbiased}, be based on a comparison of two uniform-random values.

Additionally, the comparison of any two tied events that \emph{are} related will also be based on a comparison of two-uniform random values \emph{unless} they are causally locked to a single relative ordering because one is a descendant of the other, yielding:

\begin{proposition}\label{prop:unbiased-zero-offset}
  Providing, in addition to the standard virtual timestamp, a secondary sequence of values defined by deterministic uniform random values generated by causally related zero-offset events and comparing tied events lexicographically based on that will yield an unbiased random deterministic ordering of events, including zero-offset events, without violating causality.
\end{proposition}
\begin{proof}
  The exact same logic as demonstrated in the proof of Proposition~\ref{prop:biased-zero-offset} can be applied here to guarantee determinism as lexicographical ordering provides the same strict less-than/greater-than relation between two events. Using the same events from that example, however, we can observe that should $C$ happen before $B$, then $D$ assuredly happens before $B$ as it \emph{happens infinitesimally after} $C$. The determination of the ordering of $D$ and $B$ still falls down to a fair comparison of two uniform random values: the two generated for $B$ and $C$. The probability of $D$ happening before $B$ is, effectively, randomly determined by a fair coin flip and is unbiased.
\end{proof}




\section{Experimental Models}
We evaluate our proposed solution with various methods to observe its capability as well as the performance impact of the solution's overhead. We first analyze its performance impact on a standard PDES benchmarking model: \texttt{PHOLD} with a strong-scaling study. In this model, event ties are infrequent and there are no zero-offset events.

None of the results generated from running \texttt{PHOLD}, however, will tell us about whether the ordering of events was consistent from run to run. To do that we developed a new PDES model which intentionally creates as many event ties as possible with and without zero-offset events and the final LP state is dependent on the ordering of these events. We show the performance impact in a strong scaling study using this new model.

We also stress test the simulation with an exceedingly high number of event ties in a worst-case type scenario and observe the performance impact in a third strong-scaling study based on a variation of that second model.

All experiments presented below were performed on the ROSS PDES engine with and without the unbiased tie-breaking mechanism described in Section~\ref{sec:unbiased-causal}.

\subsection{\texttt{PHOLD} Model}
We utilize the PHOLD benchmarking model included with ROSS. It is a generally well known tool for analyzing the performance of PDES systems~\cite{fujimoto1990,barnes2013}. It generates a lot of events that are sent between an LP's self and other LPs based on a configurable probability parameter.

Because of the sheer volume of events that this model can generate at large scales, this model is not immune to the possibility of event ties but is not exceedingly likely because new events are scheduled at some point strictly in the future at a time that is the result of an exponentially random number valued offset.

\subsection{\texttt{Event-Ties} Model}\label{sec:eties}
Given any implemented solution for the problems discussed in this work, it is difficult to know for certain that the determinism observed in a couple small scale experiments is true determinism or coincidence. This is especially the case if event ties are relatively infrequent.

\begin{algorithm}
  \caption{Event-Ties: Behavior of LP $P$ on receipt of event $e$}\label{alg:behavior}
  \begin{algorithmic}
    \State Given:\\\hspace{.5cm}$r$: preset threshold for desired remote events\\\hspace{.5cm}$l$: preset length of zero-offset event chain\\
    \hrulefill
    \State P.mean\_val$\leftarrow Mean($P.mean\_val, e.val$)$
    \\
    \State Create new event $n$
    \State $n$.val$\leftarrow$ Random Integer [0,100]
    \If {Random\_Unif() $< r$}
      \State \texttt{dest}$\leftarrow$random, non-self, LP 
    \Else
      \State \texttt{dest}$\leftarrow$self
    \EndIf
    \If {$e$ is the $l^{\text{th}}$ zero-offset event}
      \State Send $n$ to arrive at \texttt{dest} at time \texttt{P.now()}+1\hfill //regular-offset
    \Else
      \State Send $n$ to arrive at \texttt{dest} at time \texttt{P.now()} \hfill //zero-offset
    \EndIf
    \end{algorithmic}
\end{algorithm}

To validate that our solution accomplishes its core goal, we developed a ROSS model specifically designed to create a scenario where exact event ordering is absolutely critical to deterministic final results. A model that utilizes a running sequence of mathematical computations with strict ordering, or non-commutative operations, operating on local LP state will make non-determinism evident. Two identically configured simulations with even slightly different event ordering will yield a different final mathematical result.

For our model, we leverage the non-commutative property of the mean of means: 
\[
Mean(Mean(A,B),C) \neq Mean(Mean(A,C),B)
\]

To utilize this property, when an event is received by an LP, its state is updated to be the mean of its current state and the encoded state of the event. If two events are received by an LP at a simultaneous point in virtual time and a deterministic rule for simultaneous event ordering is not implemented, then because of Assumption~\ref{as:inter-ordering} and the above property, simulation determinism is not guaranteed.

The core behavior of the model is described in Procedure~\ref{alg:behavior}. When a new event is received, the recursive mean value in LP state is updated and a new event with a new random integer is created. Where this event is destined is based on a probability challenge so that the number of remote simulation events can be controlled. When this event is to be scheduled is determined by where in the zero-offset chain this event is. If some configured number $l-1$ same-virtual-time events causally preceded this one, then send the new event with an offset of one, thereby breaking the zero-offset chain. If this event \emph{is not} the $l^{\text{th}}$ event in the zero-offset chain, then create another zero-offset event.



\subsection{\texttt{Event-Ties-Stress} Model}
In Section~\ref{sec:eties}, we showcased an example model that would generate a large number of tied events with and without zero-offset timings. But we can make this even harder to showcase the capabilities of our tie-breaking mechanism.

Wherein Procedure~\ref{alg:behavior} creates simultaneous chains of zero-offset events, we can instead generate simultaneous trees of zero-offset events shown in Procedure~\ref{alg:behavior2}. The difficulty of this model can be scaled by adjusting the height and degree of each generated zero-offset tree. The lexicographical ordering operating on the tie-breaking sequences is particularly useful in this case as it is able to quickly determine which event happens before another, regardless of if one event is a zero-offset offspring, sibling, cousin, nephew/niece, or unrelated to the other. Furthermore, we have shown in Section~\ref{sec:unbiased-causal} that the probability of any two events being ordered a specific way in the simulation is 50\% unless one is a zero-offset descendant (child, grandchild, etc.) of the other. In that case there is a single possible relative ordering.

Also shown in Procedure~\ref{alg:behavior2} is a mechanism to make sure that only one event from a zero-offset tree creates a new zero-offset tree. Without this, there would be exponential growth of events as the simulation proceeds making it progressively more difficult to complete.

\begin{algorithm}[t]
  \caption{Event-Ties-Stress: Behavior of LP $P$ on receipt of event $e$}\label{alg:behavior2}
  \begin{algorithmic}
    \State Given:\\\hspace{.5cm}$r$: preset threshold for desired remote events\\\hspace{.5cm}$h$: preset height of zero-offset event tree\\\hspace{.5cm}$c$: preset number of zero-offset children generated per event\\
    \hrulefill
    \State P.mean\_val$\leftarrow Mean($P.mean\_val, e.val$)$
    \\
    
    \For{$i = [0:c)$}
      \State Create new event $n$
      \State $n$.val$\leftarrow$ Random Integer [0,100]
      \If {Random\_Unif() $< r$}
        \State \texttt{dest}$\leftarrow$random, non-self, LP 
      \Else
        \State \texttt{dest}$\leftarrow$self
      \EndIf
        \State$n$.descendant\_sum += $i$
        \State Send $n$ to arrive at \texttt{dest} at time \texttt{P.now()}\hfill //zero-offset
    \EndFor
    \If{$e$ is in the $h^{\text{th}}$ level of the tree \& $e$.descendant\_sum is $0$}
      \State$n$.descendant\_sum $= 0$
      \State Send $n$ to arrive at \texttt{dest} at time \texttt{P.now()}+1\hfill //regular-offset
    \EndIf
    \end{algorithmic}
\end{algorithm}

\section{Results}
All simulations performed in this work were executed using at most 4 nodes of the AiMOS supercomputer at Rensselaer Polytechnic Institute's Center for Computational Innovations~\cite{cci}. Each node contains two IBM Power 9 processors each with 20 cores with 3.15GHz clocks and 512GiB RAM. For simplicity of scaling experiments, we utilized a maximum of 32 processors per node. Inter-node communication is facilitated by an Infiniband EDR Fat Tree communication network.

Binaries were compiled using IBM's XLC\_r compiler and Spectrum MPI.

\paragraph{\texttt{PHOLD}: Strong-Scaling}\label{sec:phold-scaling}
We wanted to benchmark the performance of this extension of ROSS in direct comparison with the old version of ROSS (without the deterministic tie-breaker). The \texttt{PHOLD} model is commonly used to benchmark ROSS (and other PDES systems) and is included in the simulator's codebase~\cite{barnes2013,fujimoto1990performance}. 
These simulations were configured to each have 2 million LPs with a target of 10\% remote events generating a total of 1.3 billion net events.

What we observe in Figure~\ref{fig:phold} is that the deterministic tie-breaker version runs a bit slower at lower numbers of ranks but catches up as more ranks a re added. The overhead of implementing the tie-breaker does not significantly impact parallel performance, especially with larger numbers of ranks and because there are so few event ties, the benefit may not be obvious based on runtime alone.

\begin{figure}[t]
  \centering
  \includegraphics[width=\linewidth]{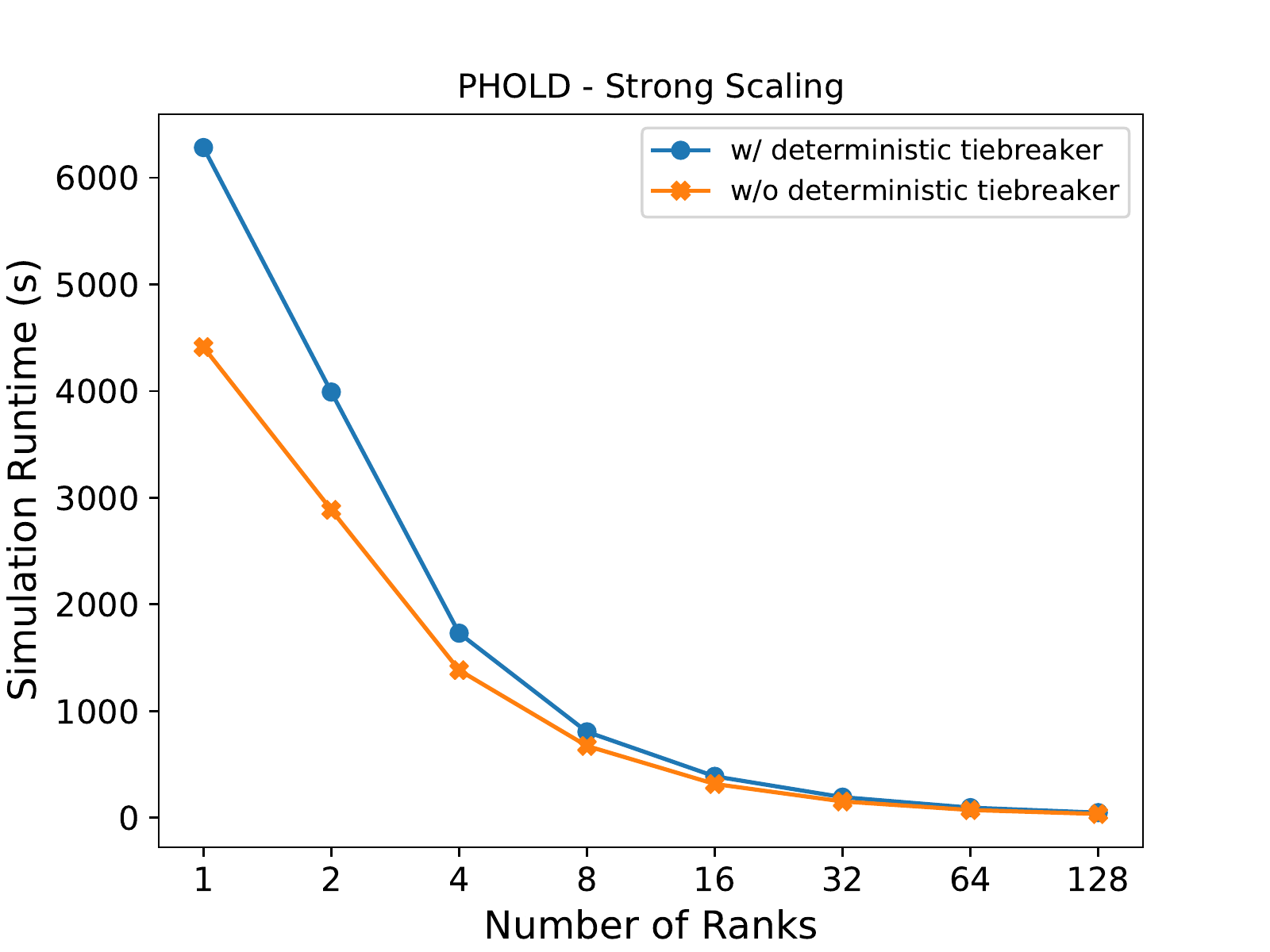}
  \caption{PHOLD strong scaling study with and without a deterministic tie-breaking mechanism.}\label{fig:phold}
\end{figure}

\paragraph{\texttt{Event-Ties}: Strong-Scaling}
Using the \texttt{Event-Ties} model, we create a situation with as many event ties as possible. All event timestamps are either zero-offset or integer incremented. 65,536 LPs were simulated, each creating a single regular-offset event which causes a single zero-offset event. In total, this created 39 million events; no event was created that was not tied with another event in the simulation.

We wanted to be able to show in a plot the capabilities of the deterministic tie-breaker version of ROSS in contrast to that of the old version. The deterministic tie-breaker version excels when faced with many zero-offset events but the old version will give the appearance of stalling as it struggles to make progress. 

The reason for this behavior is that when original ROSS instigates a rollback, it rolls back \emph{all} events with the same timestamp as the target event -- it has no way to make two events with the same timestamp comparable. This means that if a single straggler event arrives -- which it most likely will -- then any progress since the last regular offset event must be rolled back.

Because of this limitation, the only way we were able to show these two ROSS versions in a single plot with this model was to only create zero-offset chains of length two and to instead increase the target remote events to 50\%.

In Figure~\ref{fig:eties}, we can see the result of this behavior. The old ROSS version appears to do well to a point and then as more and more ranks are added, the probability of a rollback increases and eventually it takes considerable amounts of time for the simulation to complete.

\paragraph{\texttt{Event-Ties-Stress} Strong-Scaling}\label{sec:stress-scaling}
This model created zero-offset event trees instead of chains as with the standard \texttt{Event-Ties} model. This creates more zero-offset \lq siblings\rq that where the tie-breaker sequence must be fully utilized to establish a deterministic ordering.

This model was tested with 131,072 LPs each generating a 3-ary zero-offset tree of depth 5 at every integer timestep in the simulation. This created 3.9 billion total events, all with an extreme degree of ties with other events and a target remote percentage of 10\%. The small increase in runtime for the $k=2$ ranks run is because for this particular simulation model and configuration, the benefit of the additional processing power is not yet able to outweigh the cost of rollbacks which are not observed in the sequential execution.

\begin{figure}
  \centering
  \includegraphics[width=\linewidth]{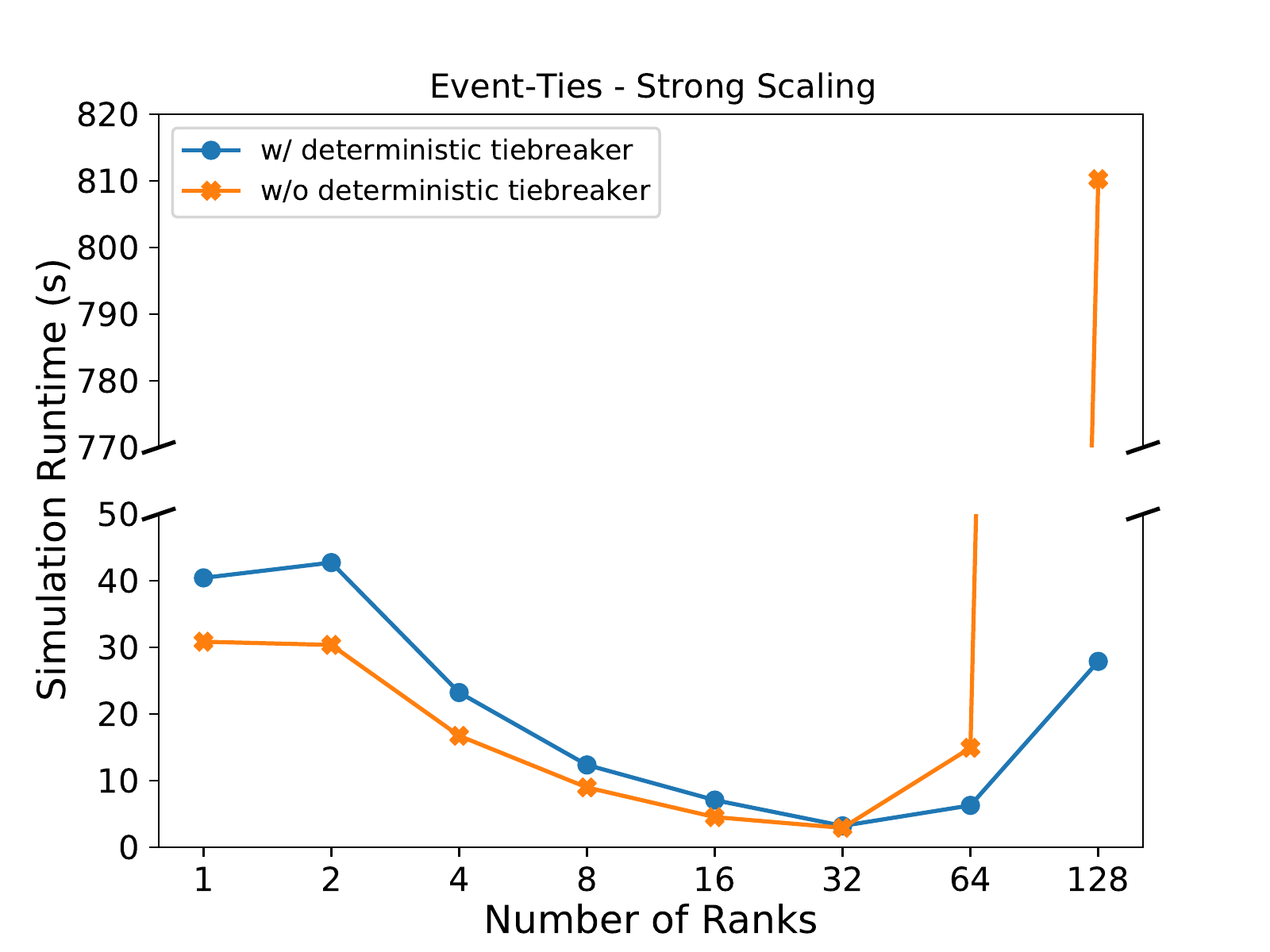}
  \caption{Event-Ties strong scaling study with and without a deterministic tie-breaking mechanism. Note that the y-axis is split to account for the last data point.} \label{fig:eties}
\end{figure}

\subsection{The Case for Determinism}
All simulations in this work utilizing the deterministic tie-breaking feature were verified to exhibit the expected determinism. In contrast, simulations executed without utilizing the tie-breaking mechanism did not consistently generate a reproducible result.

\begin{table}[b]
  \caption{Subset of final results of the sequential $k=1$, optimistic $k=32$ and $k=128$ runs for the \texttt{Event-Ties} model with and without the deterministic tie-breaker (TB)}\label{tab:final-result}
\begin{tabular}{|l|c|c|}
\hline
Simulation & Net Events & LP1 Final Value\\\hline
w/o TB: $k=1$& 39,190,528 & 35.0814\\
w/o TB: $k=32$& 39,190,528 & 34.4111\\
w/o TB: $k=128$& 39,190,528 & 34.4010\\\hline
w/ TB: $k=1$& 39,190,528 & 34.7264 \\
w/ TB: $k=32$& 39,190,528 & 34.7264 \\
w/ TB: $k=128$& 39,190,528 & 34.7264 \\
\hline
\end{tabular}
\end{table}

By looking at the results of the \texttt{PHOLD} scaling, one might argue that the overhead and performance impact is too great to warrant utilizing the feature. \texttt{PHOLD} LP state does not change throughout the simulation and the generation of future events does not depend on the contents of events received. Thus, even though event ties do occasionally occur in the simulated model, their ordering does not affect how many net events there are in the final simulation.

\begin{figure}[t]
  \centering
  \includegraphics[width=\linewidth]{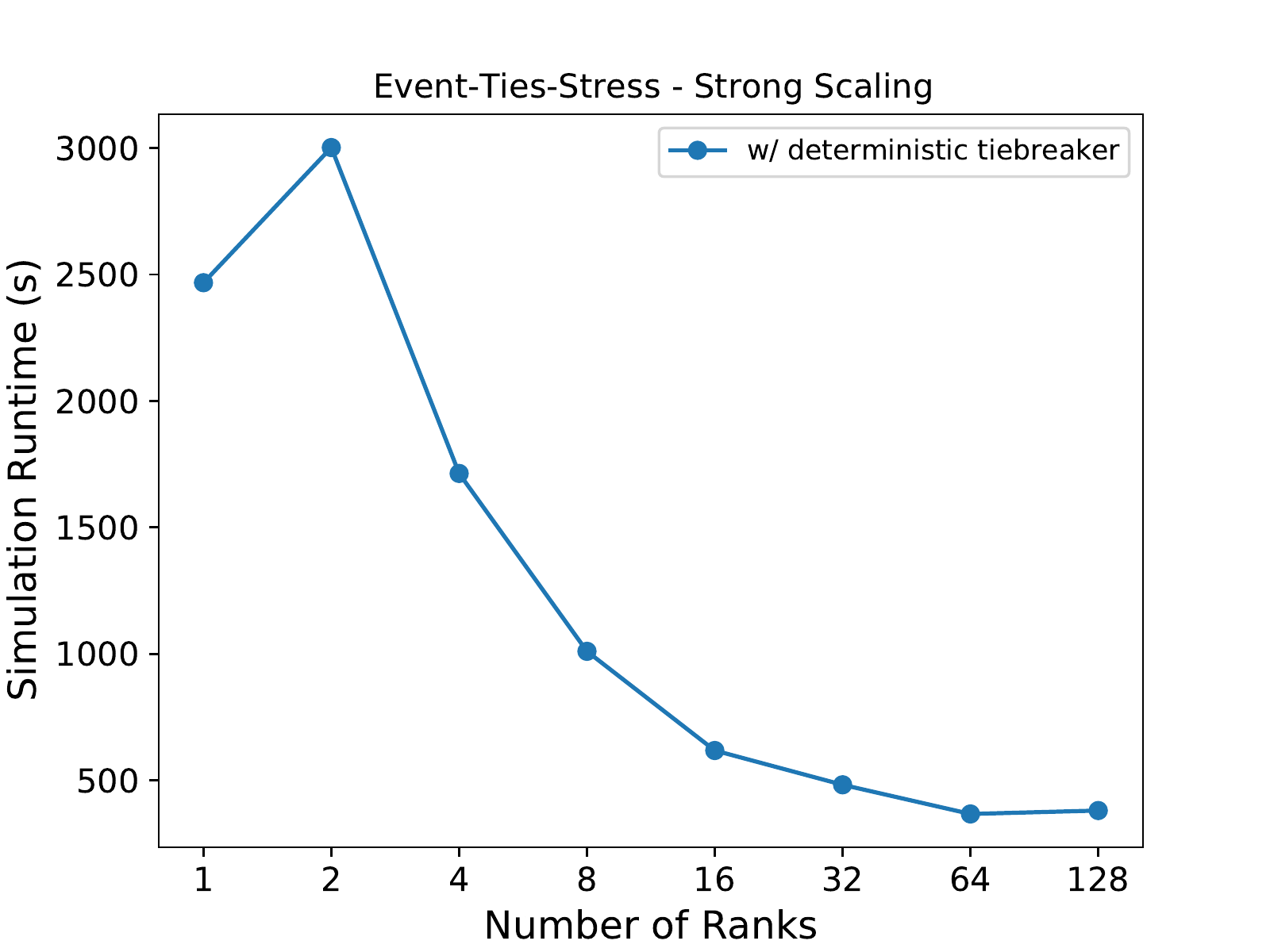}
  \caption{Event-Ties-Stress strong scaling study with deterministic tie-breaking mechanism. The old version of ROSS did not tolerate this model well and could not make noticeable progress in any execution mode but sequential.} \label{fig:eties-stress}
\end{figure}

The same cannot be said for models sensitive to simultaneous event ordering. Table~\ref{tab:final-result} shows extracted results from a few runs from the \texttt{Event-Ties} scaling results featured in Figure~\ref{fig:eties}. We observe small differences in the final results of LP state when there is no arbitration of simultaneous events -- a lack of determinism. With the deterministic tiebreaker functionality enabled, however, determinism is restored and the parallel executions are identical to their sequential counterpart. The difference between the sequential runs with and without the tiebreaker is to be expected and is because the rules dictating order in the tie-breaker simulations also affect the order of events in the sequential version.

With the \texttt{Event-Ties} model -- and \emph{any} model whose generated events depend on LP state at time of processing -- however, event ordering explicitly affects the final result. In the example shown in Table~\ref{tab:final-result} it may appear that the executions without the tie-breaker were \emph{close enough} but this is a product of the simplicity of the model and not indicative of what one could expect with other, more complicated, models. Table~\ref{tab:final-result} also shows that simply looking at the number of net events alone is not sufficient for identifying determinism in a simulated model.

If, instead, we couple the remote destination decision in Procedure~\ref{alg:behavior} to be dependent on its current LP state value, then drastically different results are observed. It makes sense that, should LP state play a role in how events are created (if they are created in the first place), then that model will be particularly sensitive to the ordering of simultaneous events. It is important to keep in mind, then -- given a standard virtual time simulation -- the level of performance observed in a model with simultaneous events may simply be borrowed with determinism offered as collateral. 

The goal of parallel discrete event simulation is to recreate a semantically identical sequential simulation but distributed across different processing elements. If determinism in the parallel execution is not assured, then one has, by definition, failed in this objective.

\section{Related Work}
The topic of event simultaneity and ordering in parallel simulations has been discussed in various contexts over the span of decades~\cite{rajaei1993,ronngren1999,kim1997,peschlow2006,peschlow2007tool,peschlow2007,lamport1978,wieland1997}. In~\cite{lamport1978}, the topic of distributed logical and physical clocks, causal ordering conditions and happens-before relations is discussed but notes that the included theory makes no requirement about the ordering of incomparable or simultaneous events. The clocks discussed base their perception of time based on the perceived happens-before relations to other events.

Virtual time~\cite{jefferson1985} takes a slightly different approach by inferring happens-before relations from the encoded timestamps giving a simulated model the power to define when events happen in relation to each other. In~\cite{schordan2020reversible}, the authors also extend ROSS' definition of virtual time to encode lower order bits to break event ties.

Much of this paper was inspired by the work performed in~\cite{wieland1997}. The author makes arguments that given a simulation with $n$ simultaneous events, then the \emph{correct} final result is the mean of all $n!$ possible event orderings. Without some mechanism to effectively explore various possible orderings a correct final result is not achievable. In our work we provide a method that allows for random sampling of the possible event orderings.

That same year, the authors of~\cite{kim1997} proposed an event ordering mechanism for simultaneous events which also employed an extension to the basic timestamp for encoding event priority. This mechanism provided identical results between a sequential and parallel executed simulation and is similar in practice to our definition of biased random causal ordering.

The authors of~\cite{ronngren1999} establish methods for breaking ties based on a logical clock interpretation of event causality or user defined priorities. They also briefly argue how different permutations of otherwise incomparable events should be considered and will inherently affect the outcomes of future events but do not elaborate on a solution to effectively explore this space. Similarly, in~\cite{peschlow2006,peschlow2007tool,peschlow2007} the authors note the importance of discovering other potential outcomes of simultaneous event orderings and propose a PDES branching mechanism tool that facilitates this.

\section{Conclusion}
In summary, we have proposed three solutions to the problem of maintaining a deterministic ordering of events in a parallel conservative or optimistically executed discrete event simulation given the possibility of $n$-way simultaneous events with and without zero-offset events. This solution leverages a deterministic pseudo-random number generator to encode values for use in breaking these $n$-way event ties into a specific, execution-consistent, ordering. This ordering can be unbiased and varied by changing the initial seed of the simulation's random number generators which allows for further statistical analysis of execution results.

We evaluated the performance impact of the proposed solution and observed a marginal impact in a standard PDES benchmarking model. In the near worst-case, stress-test, benchmark models, the performance of the simulation without the tie-breaking functionality was weakened if it was able to make effective progress at all. Furthermore, a simulation without a tie-breaking feature cannot guarantee  that it will deterministically produce an optimistic execution that is semantically identical to the sequential counterpart.

Developers and modelers should always be concerned if a program they wrote generates different results in repeated executions. Did they make a mistake somewhere? Is there a race condition they had not accounted for? We believe that determinism in parallel discrete event simulation is an important tenet and should not be cast aside in the aim of faster performance as any amount of non-determinism is, by definition, unexpected behavior.

\balance

%
\bibliographystyle{ACM-Reference-Format}
\bibliography{bibs/tiebreaker}

\end{document}